\documentclass[11pt]{article}%
\usepackage{amsmath}
\usepackage{amsfonts}
\usepackage{amssymb}
\usepackage{graphicx}
\usepackage{fullpage}
\usepackage{hyperref}%
\providecommand{\U}[1]{\protect\rule{.1in}{.1in}}
\newtheorem{theorem}{Theorem}

\newtheorem{corollary}[theorem]{Corollary}

\newtheorem{lemma}[theorem]{Lemma}

\newenvironment{proof}[1][Proof]{\noindent\textbf{#1.} }{\ \rule{0.5em}{0.5em}}
\begin{document}

\title{Quantum Lower Bound for Approximate Counting via Laurent Polynomials}
\author{Scott Aaronson\thanks{University of Texas at Austin. \ Email:
aaronson@cs.utexas.edu. \ Supported by a Vannevar Bush Fellowship from the US
Department of Defense, a Simons Investigator Award, and the Simons
\textquotedblleft It from Qubit\textquotedblright\ collaboration.}}
\date{}
\maketitle

\begin{abstract}
We consider the following problem: estimate the size of a nonempty set
$S\subseteq\left[  N\right]  $, given both quantum queries to a membership
oracle for $S$, and a device that generates equal superpositions $\left\vert
S\right\rangle $\ over $S$ elements. \ We show that, if $\left\vert
S\right\vert $\ is neither too large nor too small, then approximate counting
with these resources is still quantumly hard. \ More precisely, any quantum
algorithm needs either $\Omega\left(  \sqrt{N/\left\vert S\right\vert
}\right)  $\ queries or else $\Omega\left(  \min\left\{  \left\vert
S\right\vert ^{1/4},\sqrt{N/\left\vert S\right\vert }\right\}  \right)
$\ copies of $\left\vert S\right\rangle $. \ This means that, in the black-box
setting, quantum sampling does \textit{not} imply approximate counting. \ The
proof uses a novel generalization of the polynomial method of Beals et al. to
Laurent polynomials, which can have negative exponents.

\end{abstract}

\section{Introduction\label{INTRO}}

The \textit{quantum query complexity of approximate counting} was one of the
first topics studied in quantum algorithms. \ Given a nonempty finite set
$S\subseteq\left[  N\right]  $ (here and throughout, $\left[  N\right]
=\left\{  1,\ldots,N\right\}  $), suppose we want to estimate its cardinality,
$\left\vert S\right\vert $, to within some multiplicative accuracy
$\varepsilon$. \ This is a fundamental task in theoretical computer science,
used as a subroutine for countless other tasks.

As is standard in quantum algorithms, we work in the so-called
\textit{black-box model} (see \cite{bw}), where we assume only that we're
given a membership oracle for $S$: an oracle that, for any $i\in\left[
N\right]  $, tells us whether $i\in S$. \ We can, however, query the oracle in
quantum superposition. \ How many queries must a quantum computer make, as a
function of both $N$\ and $\left\vert S\right\vert $, to solve this problem
with high probability?

For classical randomized algorithms, one can show that $\Theta\left(  \frac
{N}{\left\vert S\right\vert }\right)  $\ membership queries are necessary and
sufficient, for approximate counting to within some constant accuracy
$\varepsilon>0$. \ Moreover, \textit{any} accuracy $\varepsilon$\ is
achievable at the cost of a\ $O\left(  \frac{1}{\varepsilon^{2}}\right)
$\ multiplicative overhead. \ Intuitively, in the worst case, we might need
$\Theta\left(  \frac{N}{\left\vert S\right\vert }\right)  $\ queries just to
find \textit{any} $S$\ elements, but once we do, estimating their frequency is
just a standard statistics problem. \ Furthermore, for the $O\left(  \frac
{N}{\left\vert S\right\vert }\right)  $ estimation strategy to work, we don't
need to suppose (circularly) that\ $\left\vert S\right\vert $\ is
approximately known in advance, but can decide when to halt dynamically,
depending on when the first $S$\ elements are found.

In the quantum case, Brassard, H\o yer, and Tapp \cite{bht:count} gave an
algorithm for approximate counting that makes only $O\left(  \sqrt{\frac
{N}{\left\vert S\right\vert }}\right)  $\ queries, for any constant
$\varepsilon>0$. \ Moreover, they showed how to achieve any
accuracy\ $\varepsilon$\ with $O\left(  \frac{1}{\varepsilon}\right)
$\ multiplicative overhead. \ To do so, one uses amplitude amplification, the
basic primitive of Grover's search algorithm \cite{grover}. \ The original
algorithm of Brassard et al. \cite{bht:count} also used quantum phase
estimation, in effect \textit{combining} Grover's algorithm with Shor's
period-finding algorithm. \ However, it's a folklore fact that one can remove
the phase estimation, and adapt Grover search with an unknown number of marked
items, to get an approximate count of the number of marked items as well.

On the lower bound side, it follows immediately from the optimality of
Grover's algorithm (i.e., the BBBV Theorem \cite{bbbv}) that even with a
quantum computer, in the black-box setting, $\Omega\left(  \sqrt{\frac
{N}{\left\vert S\right\vert }}\right)  $\ queries are \textit{needed} for
approximate counting to any constant accuracy.\bigskip

In practice, when trying to estimate the size of a set $S\subseteq\left[
N\right]  $, often we can do more than make membership queries to $S$. \ At
the least, often we can efficiently generate nearly-uniform \textit{samples}
from $S$, for instance by using Markov Chain Monte Carlo techniques. \ To give
two examples, if $S$ is the set of perfect matchings in a bipartite graph, or
the set of grid points in a high-dimensional convex body, then we can
efficiently sample $S$ using the seminal algorithms of Jerrum, Sinclair, and
Vigoda \cite{jsv} or of Dyer, Frieze, and Kannan \cite{dfk},\ respectively.

Sometimes we can even \textquotedblleft QSample\textquotedblright\ $S$---a
term coined in 2003 by Aharonov and Ta-Shma \cite{at}, and which simply means
that we can approximately prepare the uniform superposition%
\[
\left\vert S\right\rangle :=\frac{1}{\sqrt{\left\vert S\right\vert }}%
\sum_{i\in S}\left\vert i\right\rangle
\]
via a polynomial-time quantum algorithm (where \textquotedblleft
polynomial\textquotedblright\ here means $\left(  \log N\right)  ^{O\left(
1\right)  }$). \ Because we need to uncompute any garbage, the ability to
prepare $\left\vert S\right\rangle $ as a coherent superposition is a more
stringent requirement than the ability to sample $S$. \ Indeed, as Aharonov
and Ta-Shma \cite{at}\ pointed out, the quantum lower bound for finding
collisions \cite{aar:col,as} has the corollary that, in the black-box setting,
there are classes of sets $S$\ that can be efficiently sampled but
\textit{not} efficiently QSampled.

On the other hand, Aharonov and Ta-Shma \cite{at},\ and Grover and Rudolph
\cite{groverrudolph}, observed that many interesting sets $S$\ can be QSampled
as well. \ In particular, this holds for all sets $S$\ such that we can
approximately count not only $S$ itself, but also the restrictions of $S$
obtained by fixing bits of its elements. \ Or, what's known to be equivalent
\cite{sinclairjerrum}, it holds for all sets $S$ such that we can efficiently
sample not only the uniform distribution over $S$ elements, but also the
conditional distributions obtained by fixing bits. \ So in particular, the set
of perfect matchings in a bipartite graph, and the set of grid points in a
convex body, can both be efficiently QSampled. \ There are other sets that can
be QSampled but not because of this reduction. \ A simple example would be a
set $S$\ such that $\left\vert S\right\vert \geq\frac{N}%
{\operatorname*{polylog}N}$: in that case we can efficiently prepare
$\left\vert S\right\rangle $\ using postselection, but approximately counting
$S$'s restrictions might be hard.\bigskip

Recently Paul Burchard (personal communication) posed the following question
to us: are there any sets that can be QSampled even though they \textit{can't}
be approximately counted? \ Or, better: do QSampling and membership testing,
\textit{together}, imply approximate counting? \ I.e., if we have
polynomial-time quantum algorithms to prepare the state $\left\vert
S\right\rangle $, and also to decide membership in $S$, is that enough to let
us approximately count $S$?

Our main result is that, in the black-box setting, the answer to this question
is no. \ More precisely, we show that any quantum algorithm to decide whether
$\left\vert S\right\vert =w$\ or $\left\vert S\right\vert =2w$, promised that
one of those is the case, must either make $\Omega\left(  \sqrt{\frac{N}{w}%
}\right)  $\ membership queries to $S$, or else use $\Omega\left(
\min\left\{  w^{1/4},\sqrt{\frac{N}{w}}\right\}  \right)  $\ copies of
$\left\vert S\right\rangle $. \ So if (for example) we set $w:=N^{2/3}$, then
any quantum algorithm must either query $S$\ or prepare the state $\left\vert
S\right\rangle $\ at least $\Omega\left(  N^{1/6}\right)  $\ times. \ This
means that there's at most a quadratic speedup compared to classical
approximate counting.

In Section \ref{IMPROVE}, we discuss the prospects for improving this lower
bound. \ We conjecture that the bound could be improved, at least to involve
$\Omega\left(  w^{1/3}\right)  $ rather than $\Omega\left(  w^{1/4}\right)  $,
by using stronger results from approximation theory; indeed, a user on
MathOverflow already proved one of the requisite results about polynomial
degree. \ However, we also observe that our lower bound \textit{cannot} be
improved to involve $\omega\left(  w^{1/3}\right)  $\ without going beyond the
polynomial method. \ While we do not do that in this paper, we give a viable
approach:\ using a hybrid argument, inspired by recent work of Zhandry
\cite{zhandry:lightning}, we show that better lower bounds for our problem
would follow from better lower bounds purely on the number of copies of
$\left\vert S\right\rangle $ (ignoring the number of queries).

Our lower bounds are within a polynomial factor of the best known quantum
upper bounds for approximate counting. \ As mentioned before, Brassard et al.
\cite{bht:count} gave a quantum algorithm to solve the problem using $O\left(
\sqrt{\frac{N}{w}}\right)  $\ queries (and no copies of $\left\vert
S\right\rangle $). \ At the opposite extreme, it's easy to solve the problem
using $O\left(  \sqrt{w}\right)  $\ copies of $\left\vert S\right\rangle $
(and no queries), by simply measuring each copy of $\left\vert S\right\rangle
$\ in the computational basis and then searching for birthday collisions.
\ Alternatively, one can solve the problem using $O\left(  \frac{N}{w}\right)
$\ copies of $\left\vert S\right\rangle $ (and again, no queries), by
projecting each copy onto the state $\frac{1}{\sqrt{N}}\left(  \left\vert
1\right\rangle +\cdots+\left\vert N\right\rangle \right)  $\ and then counting
how many of the projections succeed. \ We're not aware of any quantum
algorithm for approximate counting that combines membership queries with
QSampling in an interesting way, though neither can we rule such an algorithm
out. \ Of course, our main result limits the advantage that any such algorithm
could achieve.

In our view, at least as interesting as our main result is the technique used
to achieve it. \ In 1998, Beals et al. \cite{bbcmw}\ famously observed that,
if a quantum algorithm $Q$\ makes $T$ queries to an input $X$, then $Q$'s
acceptance probability can be written as a real multilinear polynomial in the
bits of $X$, of degree at most $2T$. \ And thus, crucially, if we want to
\textit{rule out} a fast quantum algorithm to compute some function $f\left(
X\right)  $, then it suffices to show that any real polynomial $p$\ that
approximates $f$\ pointwise must have high degree. \ This general
transformation, from questions about quantum algorithms to questions about
polynomials, has been used to prove many results that were not known otherwise
at the time, including the quantum lower bound for the collision problem
\cite{aar:col,as}\ and the first direct product theorems for quantum search
\cite{aar:adv,ksw}.

In our case, the difficulty is that the quantum algorithm starts with many
copies of the state $\left\vert S\right\rangle $. \ As a consequence of
this---and specifically, of the $\frac{1}{\sqrt{\left\vert S\right\vert }}%
$\ normalizing factor in $\left\vert S\right\rangle $---when we write the
average acceptance probability of our algorithm as a function of $\left\vert
S\right\vert $, we find that we get a \textit{Laurent polynomial}: a
polynomial that can contain both positive and negative integer powers of
$\left\vert S\right\vert $. \ The degree of this polynomial (the highest power
of $\left\vert S\right\vert $) encodes the sum of the number of queries and
the number of copies of $\left\vert S\right\rangle $, while the
\textquotedblleft anti-degree\textquotedblright\ (the highest power of
$\frac{1}{\left\vert S\right\vert }$) encodes the number of copies of
$\left\vert S\right\rangle $. \ We're thus faced with the task of
lower-bounding the degree and the anti-degree of a Laurent polynomial that's
bounded at integer points and that encodes the approximate counting problem.

We address this using a switching argument that, as far as we know, is new in
quantum query complexity. \ Writing our Laurent polynomial as $q\left(
\left\vert S\right\vert \right)  =u\left(  \left\vert S\right\vert \right)
+v(\frac{1}{\left\vert S\right\vert })$, where $u$ and $v$ are ordinary
polynomials, we show that, if $u$ and $v$ both have low enough degree (namely,
$\deg\left(  u\right)  =o\left(  \sqrt{\frac{N}{w}}\right)  $ and $\deg\left(
v\right)  =o\left(  w^{1/4}\right)  $), then we get \textquotedblleft
unbounded growth\textquotedblright\ in their values. \ That is: for
approximation theory reasons, either $u$ or $v$ must attain large values, far
outside of $\left[  0,1\right]  $, at some integer values of $\left\vert
S\right\vert $. \ But that means that, for $q$ itself to be bounded in
$\left[  0,1\right]  $\ (and thus represent a probability), the other
polynomial must \textit{also} attain large values. \ And that, in turn, will
force the first polynomial to attain even larger values, and so on
forever---thereby proving that these polynomials could not have existed.

We observe that, if we considered the broader class of rational functions,
then there \textit{are} rational functions of low degree that represent
approximate counting. \ This follows, for example, from the connection between
rational functions and \textit{postselected} quantum algorithms
\cite{mahadevdewolf}, together with Aaronson's $\mathsf{PostBQP}=\mathsf{PP}%
$\ theorem \cite{aar:pp}. \ Thus, our proof relies on the fact that Laurent
polynomials are an extremely special kind of rational function.

\section{Result\label{RESULT}}

Define $x\in\left\{  0,1\right\}  ^{N}$\ , the \textquotedblleft
characteristic string\textquotedblright\ of the set $S\subseteq\left[
N\right]  $, by $x_{i}=1$\ if $i\in S$\ and $x_{i}=0$\ otherwise.

Our starting point is the well-known \textit{symmetrization lemma} of Minsky
and Papert \cite{mp} (see also Beals et al. \cite{bbcmw}\ for its application
to quantum query complexity), by which we can often reduce questions about
multivariate polynomials to questions about univariate ones.

\begin{lemma}
[Symmetrization Lemma, Minsky and Papert \cite{mp}]\label{symlem}Let
$p:\left\{  0,1\right\}  ^{N}\rightarrow\mathbb{R}$\ be a real multilinear
polynomial of degree $d$, and let%
\[
q\left(  k\right)  :=\operatorname{E}_{\left\vert X\right\vert =k}\left[
p\left(  X\right)  \right]  .
\]
Then $q$ can be written as a real polynomial in $k$ of degree at most $d$.
\end{lemma}

By using Lemma \ref{symlem}, we now prove the key fact that relates quantum
algorithms, of the type we're considering, to real Laurent polynomials in one
variable. \ The following lemma generalizes the connection between quantum
algorithms and real polynomials established by Beals et al. \cite{bbcmw}.

\begin{lemma}
\label{laurentlem}Let $Q$ be a quantum algorithm that receives $R$\ copies of
$\left\vert S\right\rangle $ and makes $T$ queries to $\mathcal{O}_{S}$. \ Let%
\[
q\left(  k\right)  :=\operatorname{E}_{\left\vert S\right\vert =k}\left[
\Pr\left[  Q^{\mathcal{O}_{S}}\left(  \left\vert S\right\rangle ^{\otimes
R}\right)  \text{ accepts}\right]  \right]  .
\]
Then $q:\mathbb{R}\rightarrow\mathbb{R}$\ is a univariate Laurent polynomial,
with maximum exponent at most $2T+R$\ and minimum exponent at least $-R$.
\end{lemma}

\begin{proof}
Let $\left\vert \psi_{t}\right\rangle $ be $Q$'s state immediately after the
$t^{th}$\ query. \ Then we can write $Q$'s initial state as%
\[
\left\vert \psi_{0}\right\rangle =\left\vert S\right\rangle ^{\otimes R}%
=\frac{1}{\left\vert S\right\vert ^{R/2}}\sum_{i_{1},\ldots,i_{R}\in\left[
N\right]  }x_{i_{1}}\cdots x_{i_{R}}\left\vert i_{1},\ldots,i_{R}\right\rangle
.
\]
Thus, each amplitude is a complex multilinear polynomial in $X=\left(
x_{1},\ldots,x_{N}\right)  $ of degree $R$, divided by $\left\vert
S\right\vert ^{R/2}$.

Like Beals et al. \cite{bbcmw}, we now consider how amplitudes change as $Q$
progresses. \ Each query, to an index $i\in\left[  N\right]  $, multiplies the
amplitude of the associated basis state by $1-2x_{i}$, increasing the
amplitude's degree as a polynomial by $1$. \ Meanwhile, between the $t^{th}%
$\ and $\left(  t+1\right)  ^{st}$\ queries, $Q$ can apply an arbitrary
unitary transformation $U_{t}$, which does not depend on $X$ and hence does
not increase degree. \ Since $x_{i}^{2}=x_{i}$\ for all $i$, we can also
maintain multilinearity without loss of generality.

It follows that $Q$'s final state has the form%
\[
U_{T}\left\vert \psi_{T}\right\rangle =\sum\alpha_{z}\left(  X\right)
\left\vert z\right\rangle ,
\]
where each $\alpha_{z}\left(  X\right)  $\ is a complex multilinear polynomial
in $X$ of degree at most $R+T$, again divided by $\left\vert S\right\vert
^{R/2}$. \ Since $X$ itself is real-valued, it follows that the real and
imaginary parts of $\alpha_{z}\left(  X\right)  $, considered individually,
are real multilinear polynomials in $X$\ of degree at most $R+T$\ divided by
$\left\vert S\right\vert ^{R/2}$.

Hence, if we let%
\[
p\left(  X\right)  :=\Pr\left[  Q^{\mathcal{O}_{S}}\left(  \left\vert
S\right\rangle ^{\otimes R}\right)  \text{ accepts}\right]  ,
\]
then%
\[
p\left(  X\right)  =\sum_{\text{accepting }z}\left\vert \alpha_{z}\left(
X\right)  \right\vert ^{2}=\sum_{\text{accepting }z}\left(  \operatorname{Re}%
^{2}\alpha_{z}\left(  X\right)  +\operatorname{Im}^{2}\alpha_{z}\left(
X\right)  \right)
\]
is a real multilinear polynomial in $X$ of degree at most $2\left(
R+T\right)  $, divided through (in every monomial) by $\left\vert S\right\vert
^{R}=\left\vert X\right\vert ^{R}$.

Now consider%
\[
q\left(  k\right)  :=\operatorname{E}_{\left\vert X\right\vert =k}\left[
p\left(  X\right)  \right]  .
\]
By Lemma \ref{symlem}, this is a real univariate polynomial in $\left\vert
X\right\vert $ of degree at most $2\left(  R+T\right)  $, divided through (in
every monomial) by $\left\vert S\right\vert ^{R}=\left\vert X\right\vert ^{R}%
$. \ Or said another way, it's a real Laurent polynomial in $\left\vert
X\right\vert $, with maximum exponent at most $R+2T$\ and minimum exponent at
least $-R$.
\end{proof}

Besides , We'll need several results from approximation theory,\ each of which
has previously been used (in some form) in other applications of the
polynomial method to quantum lower bounds. \ We start with the basic
inequality of Markov.

\begin{lemma}
[Markov]\label{markovlem}Let $p$\ be a real polynomial, and suppose that%
\[
\max_{x,y\in\left[  a,b\right]  }\left\vert p\left(  x\right)  -p\left(
y\right)  \right\vert \leq H.
\]
Then%
\[
\left\vert p^{\prime}\left(  x\right)  \right\vert \leq\frac{H}{b-a}%
\deg\left(  p\right)  ^{2}%
\]
for all $x\in\left[  a,b\right]  $.
\end{lemma}

We'll also need a bound that was explicitly stated by Paturi \cite{paturi},
and which amounts to the folklore fact that, among all degree-$d$ polynomials
that are bounded within a given range, the Chebyshev polynomials have the
fastest growth outside that range.

\begin{lemma}
[Paturi]\label{paturilem}Let $p$\ be a real polynomial, and suppose that
$\left\vert p\left(  x\right)  \right\vert \leq1$\ for all $\left\vert
x\right\vert \leq1$. \ Then for all $x\leq1+\mu$, we have%
\[
\left\vert p\left(  x\right)  \right\vert \leq\exp\left(  2\deg\left(
p\right)  \sqrt{2\mu+\mu^{2}}\right)  .
\]

\end{lemma}

We now state a useful corollary of Lemma \ref{paturilem}, which says (in
effect) that slightly shrinking the domain of a low-degree real polynomial can
only modestly shrink its range.

\begin{corollary}
\label{paturicor}Let $p$\ be a real polynomial of degree $d$, and suppose that%
\[
\max_{x,y\in\left[  a,b\right]  }\left\vert p\left(  x\right)  -p\left(
y\right)  \right\vert \geq H.
\]
Let $\varepsilon\leq\frac{1}{100d^{2}}$\ and $a^{\prime}:=a+\varepsilon\left(
b-a\right)  $. \ Then%
\[
\max_{x,y\in\left[  a^{\prime},b\right]  }\left\vert p\left(  x\right)
-p\left(  y\right)  \right\vert \geq\frac{H}{2}.
\]

\end{corollary}

\begin{proof}
Suppose by contradiction that%
\[
\left\vert p\left(  x\right)  -p\left(  y\right)  \right\vert <\frac{H}{2}%
\]
for all $x,y\in\left[  a^{\prime},b\right]  $. \ By affine shifts, we can
assume without loss of generality that $\left\vert p\left(  x\right)
\right\vert <\frac{H}{4}$\ for all $x\in\left[  a^{\prime},b\right]  $. \ Then
by Lemma \ref{paturilem}, for all $x\in\left[  a,b\right]  $\ we have%
\[
\left\vert p\left(  x\right)  \right\vert <\frac{H}{4}\cdot\exp\left(
2d\sqrt{2\left(  \frac{1}{1-\varepsilon}-1\right)  +\left(  \frac
{1}{1-\varepsilon}-1\right)  ^{2}}\right)  \leq\frac{H}{2}.
\]
But this violates the hypothesis.
\end{proof}

Finally, we'll need a bound that relates the range of a low-degree polynomial
on a discrete set of points to its range on a continuous interval. \ The
following lemma generalizes a result due to Ehlich and Zeller \cite{ez} and
Rivlin and Cheney \cite{rc}, who were interested only in the case where the
discrete points are evenly spaced.

\begin{lemma}
\label{ezrclem}Let $p$\ be a real polynomial of degree at most $\sqrt{N}$, and
let $0=z_{1}<\cdots<z_{M}=N$\ be a list of points such that $z_{i+1}-z_{i}%
\leq1$\ for all $i$ (the simplest example being the integers $0,\ldots,N$).
\ Suppose that%
\[
\max_{x,y\in\left[  0,N\right]  }\left\vert p\left(  x\right)  -p\left(
y\right)  \right\vert \geq H.
\]
Then%
\[
\max_{i,j}\left\vert p\left(  z_{i}\right)  -p\left(  z_{j}\right)
\right\vert \geq\frac{H}{2}.
\]

\end{lemma}

\begin{proof}
Suppose by contradiction that%
\[
\left\vert p\left(  z_{i}\right)  -p\left(  z_{j}\right)  \right\vert
<\frac{H}{2}%
\]
for all $i,j$. \ By affine shifts, we can assume without loss of generality
that $\left\vert p\left(  z_{i}\right)  \right\vert <\frac{H}{4}$\ for all
$i$. \ Let%
\[
c:=\max_{x\in\left[  0,N\right]  }\frac{\left\vert p\left(  x\right)
\right\vert }{H/4}.
\]
If $c\leq1$, then the hypothesis clearly fails, so assume $c>1$. \ Suppose
that the maximum, $\left\vert p\left(  x\right)  \right\vert =\frac{cH}{4}$,
is achieved between $z_{i}$\ and $z_{i+1}$. \ Then by basic calculus, there
exists an $x^{\ast}\in\left[  z_{i},z_{i+1}\right]  $\ such that%
\[
\left\vert p^{\prime}\left(  x^{\ast}\right)  \right\vert >\frac{2\left(
c-1\right)  }{z_{i+1}-z_{i}}\cdot\frac{H}{4}\geq\frac{\left(  c-1\right)
H}{2}.
\]
So by Lemma \ref{markovlem},%
\[
\frac{\left(  c-1\right)  H}{2}<\frac{cH/4}{N}\deg\left(  p\right)  ^{2}.
\]
Solving for $c$, we find%
\[
c<\frac{2N}{2N-\deg\left(  p\right)  ^{2}}\leq2.
\]
But if $c<2$, then $\max_{x\in\left[  0,N\right]  }\left\vert p\left(
x\right)  \right\vert <\frac{H}{2}$, which violates the hypothesis.
\end{proof}

We're now ready to prove the main result of this paper.

\begin{theorem}
\label{main}Let $Q$ be a quantum algorithm that receives $R$\ copies of
$\left\vert S\right\rangle $, makes $T$ queries to $\mathcal{O}_{S}$, and
decides whether $\left\vert S\right\vert =w$\ or $\left\vert S\right\vert =2w$
with success probability at least $2/3$, promised that one of those is the
case. \ Then either $T=\Omega\left(  \sqrt{\frac{N}{w}}\right)  $\ or%
\[
R=\Omega\left(  \min\left\{  w^{1/4},\sqrt{\frac{N}{w}}\right\}  \right)  .
\]

\end{theorem}

\begin{proof}
Let%
\[
q\left(  k\right)  :=\operatorname{E}_{\left\vert S\right\vert =k}\left[
\Pr\left[  Q^{\mathcal{O}_{S}}\left(  \left\vert S\right\rangle ^{\otimes
R}\right)  \text{ accepts}\right]  \right]  .
\]
Then by Lemma \ref{laurentlem}, we can write $q$ as a Laurent polynomial, like
so:%
\[
q\left(  k\right)  =u\left(  k\right)  +v\left(  1/k\right)  ,
\]
where $u$\ is a real polynomial in $k$ with $\deg\left(  u\right)  \leq2T+R$,
and $v$\ is a real polynomial in $1/k$\ with $\deg\left(  v\right)  \leq R$.
\ So to prove the theorem, it suffices to show that either $\deg\left(
u\right)  =\Omega\left(  \sqrt{\frac{N}{w}}\right)  $, or else $\deg\left(
v\right)  =\Omega\left(  w^{1/4}\right)  $. \ To do so, we'll assume that
$\deg\left(  u\right)  =o\left(  \sqrt{\frac{N}{w}}\right)  $\ and
$\deg\left(  v\right)  =o\left(  w^{1/4}\right)  $, and derive a contradiction.

Our high-level strategy is as follows: we'll observe that, if approximate
counting is successfully being solved, then either $u$ or $v$ must attain a
large first derivative somewhere in its domain. \ By the approximation theory
lemmas that we proved earlier, this will force that polynomial to have a large
range---even on a subset of integer (or inverse-integer) points. \ But the
sum, $u\left(  k\right)  +v\left(  1/k\right)  $, is bounded in $\left[
0,1\right]  $\ for all $k\in\left[  N\right]  $. \ So if one polynomial has a
large range, then the other does too. \ But this forces the \textit{other}
polynomial to have a large derivative somewhere in its domain, and therefore
(by approximation theory) to have an even larger range, forcing the first
polynomial to have an even larger range to compensate, and so on. \ As long as
$\deg\left(  u\right)  $\ and $\deg\left(  v\right)  $ are both small enough,
this endless switching will force both $u$ and $v$ to attain
\textit{unboundedly }large values---with the fact that one polynomial is in
$k$, and the other is in $1/k$, crucial to achieving the desired
\textquotedblleft explosion.\textquotedblright\ \ Since $u$ and $v$ are
polynomials on compact sets, such unbounded growth is an obvious absurdity,
and this will give us the desired contradiction.

In more detail, we will study the following quantities.%
\[%
\begin{tabular}
[c]{ll}%
$G_{u}:=\max_{x,y\in\left[  \sqrt{w},2w\right]  }\left\vert u\left(  x\right)
-u\left(  y\right)  \right\vert ~\ \ \ \ \ \ $ & $G_{v}:=\max_{x,y\in\left[
\frac{1}{N},\frac{1}{w}\right]  }\left\vert v\left(  x\right)  -v\left(
y\right)  \right\vert $\\
$\Delta_{u}:=\max_{x\in\left[  \sqrt{w},2w\right]  }\left\vert u^{\prime
}\left(  x\right)  \right\vert $ & $\Delta_{v}:=\max_{x\in\left[  \frac{1}%
{N},\frac{1}{w}\right]  }\left\vert v^{\prime}\left(  x\right)  \right\vert
$\\
$H_{u}:=\max_{x,y\in\left[  \sqrt{w},N\right]  }\left\vert u\left(  x\right)
-u\left(  y\right)  \right\vert $ & $H_{v}:=\max_{x,y\in\left[  \frac{1}%
{N},\frac{1}{\sqrt{w}}\right]  }\left\vert v\left(  x\right)  -v\left(
y\right)  \right\vert $\\
$I_{u}:=\max_{x,y\in\left[  w,N\right]  }\left\vert u\left(  x\right)
-u\left(  y\right)  \right\vert $ & $I_{v}:=\max_{x,y\in\left[  \frac{1}%
{2w},\frac{1}{\sqrt{w}}\right]  }\left\vert v\left(  x\right)  -v\left(
y\right)  \right\vert $\\
$L_{u}:=\max_{x,y\in\left\{  w,\ldots,N\right\}  }\left\vert u\left(
x\right)  -u\left(  y\right)  \right\vert $ & $L_{v}:=\max_{x,y\in\left\{
\sqrt{w},\ldots,2w\right\}  }\left\vert v\left(  \frac{1}{x}\right)  -v\left(
\frac{1}{y}\right)  \right\vert $%
\end{tabular}
\]

We have $0\leq q\left(  k\right)  \leq1$\ for all $k\in\left[  N\right]  $,
since in those cases $q\left(  k\right)  $\ represents a probability. \ Since
$Q$ solves approximate counting, we also have $q\left(  w\right)  \leq\frac
{1}{3}$\ and $q\left(  2w\right)  \geq\frac{2}{3}$. \ This means in particular
that either

\begin{enumerate}
\item[(i)] $u\left(  2w\right)  -u\left(  w\right)  \geq\frac{1}{6}$, and
hence $G_{u}\geq\frac{1}{6}$, or else

\item[(ii)] $v\left(  \frac{1}{2w}\right)  -v\left(  \frac{1}{w}\right)
\geq\frac{1}{6}$, and hence $G_{v}\geq\frac{1}{6}$.
\end{enumerate}

We will show that either case leads to a contradiction.

We have the following inequalities regarding $u$:%
\[%
\begin{tabular}
[c]{ll}%
$G_{u}\geq L_{v}-1$ & by the boundedness of $q$\\
$\Delta_{u}\geq\frac{G_{u}}{2w}$ & by basic calculus\\
$H_{u}\geq\frac{\Delta_{u}\left(  N-\sqrt{w}\right)  }{\deg\left(  u\right)
^{2}}$ & by Lemma \ref{markovlem}\\
$I_{u}\geq\frac{H_{u}}{2}$ & by Corollary \ref{paturicor}\\
$L_{u}\geq\frac{I_{u}}{2}$ & by Lemma \ref{ezrclem}%
\end{tabular}
\]
Here the fourth inequality uses the fact that, setting $\varepsilon
:=\frac{\sqrt{w}}{N}$, we have $\deg\left(  u\right)  =o\left(  \frac{1}%
{\sqrt{\varepsilon}}\right)  $ (thereby satisfying the hypothesis of Corollary
\ref{paturicor}), while the fifth inequality uses the fact that $\deg\left(
u\right)  =o\left(  \sqrt{N}\right)  $.

Meanwhile, we have the following inequalities regarding $v$:%
\[%
\begin{tabular}
[c]{ll}%
$G_{v}\geq L_{u}-1$ & by the boundedness of $q$\\
$\Delta_{v}\geq G_{v}w$ & by basic calculus\\
$H_{v}\geq\frac{\Delta_{v}\left(  \frac{1}{\sqrt{w}}-\frac{1}{N}\right)
}{\deg\left(  v\right)  ^{2}}$ & by Lemma \ref{markovlem}\\
$I_{v}\geq\frac{H_{v}}{2}$ & by Corollary \ref{paturicor}\\
$L_{v}\geq\frac{I_{v}}{2}$ & by Lemma \ref{ezrclem}%
\end{tabular}
\]
Here the fourth inequality uses the fact that, setting $\varepsilon
:=\frac{1/2w}{1/\sqrt{w}}=\frac{1}{2\sqrt{w}}$, we have $\deg\left(  v\right)
=o\left(  \frac{1}{\sqrt{\varepsilon}}\right)  $ (thereby satisfying the
hypothesis of Corollary \ref{paturicor}). \ The fifth inequality uses the fact
that, if we set $V\left(  x\right)  :=v\left(  x/w\right)  $, then the
situation\ satisfies the hypothesis of Lemma \ref{ezrclem}:\ we are interested
in the range of $V$ on the interval $\left[  \frac{1}{2},\sqrt{w}\right]  $,
compared to its range on discrete points $\frac{w}{\sqrt{w}},\frac{w}{\sqrt
{w}+1},\ldots,\frac{w}{2w}$\ that are spaced at most $1$ apart from each
other; and we also have $\deg\left(  V\right)  =\deg\left(  v\right)
=o\left(  w^{1/4}\right)  $.

All that remains is to show that, if we insert either $G_{u}\geq\frac{1}{6}%
$\ or $G_{v}\geq\frac{1}{6}$ into the coupled system of inequalities above,
then we get unbounded growth and the inequalities have no solution. \ Let us
collapse the two sets of inequalities to%
\begin{align*}
L_{u}  &  \geq\frac{1}{4}\frac{N-\sqrt{w}}{\deg\left(  u\right)  ^{2}}%
\frac{G_{u}}{2w}=\Omega\left(  \frac{N}{w\deg\left(  u\right)  ^{2}}%
G_{u}\right)  ,\\
L_{v}  &  \geq\frac{1}{4}\frac{\frac{1}{\sqrt{w}}-\frac{1}{N}}{\deg\left(
v\right)  ^{2}}G_{v}w=\Omega\left(  \frac{\sqrt{w}}{\deg\left(  v\right)
^{2}}G_{v}\right)  .
\end{align*}
Hence%
\begin{align*}
G_{u}  &  \geq L_{v}-1=\Omega\left(  \frac{\sqrt{w}}{\deg\left(  v\right)
^{2}}G_{v}\right)  -1,\\
G_{v}  &  \geq L_{u}-1=\Omega\left(  \frac{N}{w\deg\left(  u\right)  ^{2}%
}G_{u}\right)  -1.
\end{align*}
By the assumption that $\deg\left(  v\right)  =o\left(  w^{1/4}\right)  $\ and
$\deg\left(  u\right)  =o\left(  \sqrt{\frac{N}{w}}\right)  $, we have
$\frac{\sqrt{w}}{\deg\left(  v\right)  ^{2}}\gg1$\ and $\frac{N}{w\deg\left(
u\right)  ^{2}}\gg1$. \ Plugging in $G_{u}\geq\frac{1}{6}$\ or $G_{v}\geq
\frac{1}{6}$, this is enough to give us unbounded growth.
\end{proof}

\section{Improvements\label{IMPROVE}}

At our request, user \textquotedblleft fedja\textquotedblright\ on
MathOverflow kindly proved the following lemma in approximation theory (see
the link\footnote{See
https://mathoverflow.net/questions/302113/real-polynomial-bounded-at-inverse-integer-points}
for the proof):

\begin{lemma}
[fedja]\label{fedjalem}Let $p$\ be a real polynomial, and suppose that
$\left\vert p\left(  1/k\right)  \right\vert \leq1$ for all $k\in\left[
2w\right]  $, and that $p\left(  \frac{1}{w}\right)  \leq\frac{1}{3}$\ while
$p\left(  \frac{1}{2w}\right)  \geq\frac{2}{3}$. \ Then $\deg\left(  p\right)
=\Omega\left(  w^{1/3}\right)  $.
\end{lemma}

Interestingly, Lemma \ref{fedjalem} turns out to be tight. \ We give the
construction for completeness:

\begin{lemma}
[fedja]\label{fedjatight}For all $w$, there is a real polynomial $p$\ such
that $\left\vert p\left(  1/k\right)  \right\vert \leq1$ for all $k\in\left[
2w\right]  $, and $p\left(  \frac{1}{w}\right)  \leq\frac{1}{3}$\ while
$p\left(  \frac{1}{2w}\right)  \geq\frac{2}{3}$, and $\deg\left(  p\right)
=O\left(  w^{1/3}\right)  $.
\end{lemma}

\begin{proof}
Assuming for simplicity that $w$ is a perfect cube, consider%
\[
u\left(  x\right)  :=\left(  1-x\right)  \left(  1-2x\right)  \cdots\left(
1-w^{1/3}x\right)  .
\]
Notice that $\deg\left(  u\right)  =w^{1/3}$\ and $u\left(  \frac{1}%
{k}\right)  =0$ for all $k\in\left[  w^{1/3}\right]  $. \ Furthermore, we have
$\left\vert u\left(  x\right)  \right\vert \leq1$\ for all $x\in\left[
0,\frac{1}{w^{1/3}}\right]  $, and also $u\left(  x\right)  \in\left[
1-O\left(  \frac{1}{w^{1/3}}\right)  ,1\right]  $\ for all $x\in\left[
0,\frac{1}{w}\right]  $. \ Now, let $v$\ be the Chebyshev polynomial of degree
$w^{1/3}$, affinely adjusted so that $\left\vert v\left(  x\right)
\right\vert \leq1$\ for all $x\in\left[  0,\frac{1}{w^{1/3}}\right]  $ rather
than all $\left\vert x\right\vert \leq1$, and with a large jump between
$\frac{1}{2w}$\ and $\frac{1}{w}$. \ Then the product, $p\left(  x\right)
:=u\left(  x\right)  v\left(  x\right)  $, has degree $2w^{1/3}$\ and
satisfies all the requirements.
\end{proof}

It seems plausible that, by using Lemma \ref{fedjalem}, we could give a modest
improvement to Theorem \ref{main}, which would involve $\Omega\left(
w^{1/3}\right)  $ rather than $\Omega\left(  w^{1/4}\right)  $.
\ Unfortunately, there are technical difficulties in doing so, since relaxing
the assumption\ $\deg\left(  v\right)  =o\left(  w^{1/4}\right)  $\ to
$\deg\left(  v\right)  =o\left(  w^{1/3}\right)  $ breaks several steps in the
proof simultaneously. \ We leave the details to future work.

In any case, Lemma \ref{fedjatight}\ presumably means that, to prove a lower
bound involving $\Omega\left(  \sqrt{w}\right)  $, one would need to go beyond
the polynomial method. \ (We say \textquotedblleft
presumably\textquotedblright\ because we can't rule out the possibility of
using the polynomial method in some way completely different from how it was
used in Theorem \ref{main}.)\newline

In the remainder of this section, we give what we think is a viable path to
going beyond the polynomial method. \ Specifically, we observe that our
problem---of lower-bounding the number of copies of $\left\vert S\right\rangle
$\ \textit{and} the number of queries to $\mathcal{O}_{S}$\ needed for
approximate counting of $S$---can be reduced to a pure problem of
lower-bounding the number of copies of $\left\vert S\right\rangle $. \ To do
so, we use a hybrid argument, closely analogous to an argument recently given
by Zhandry \cite{zhandry:lightning} in the context of quantum money.

Given a subset $S\subseteq\left[  L\right]  $, let $\left\vert S\right\rangle
$\ be a uniform superposition over $S$ elements. \ Then let%
\[
\rho_{L,w,k}:=\operatorname{E}_{S\subseteq\left[  L\right]  ~:~\left\vert
S\right\vert =w}\left[  \left(  \left\vert S\right\rangle \left\langle
S\right\vert \right)  ^{\otimes k}\right]
\]
be the mixed state obtained by first choosing $S$\ uniformly at random subject
to $\left\vert S\right\vert =w$, then taking $k$ copies of $\left\vert
S\right\rangle $. \ Given two mixed states $\rho$\ and $\sigma$, recall also
that the \textit{trace distance}, $\left\Vert \rho-\sigma\right\Vert
_{\operatorname*{tr}}$, is the maximum bias with which $\rho$\ can be
distinguished from $\sigma$\ by a single-shot measurement.

\begin{theorem}
\label{hybridthm}Let $2w\leq L\leq N$. \ Suppose $\left\Vert \rho_{L,w,k}%
-\rho_{L,2w,k}\right\Vert _{\operatorname*{tr}}\leq\frac{1}{10}$. \ Then any
quantum algorithm $Q$\ requires either $\Omega\left(  \sqrt{\frac{N}{L}%
}\right)  $\ queries to $\mathcal{O}_{S}$\ or else $\Omega\left(  k\right)
$\ copies of $\left\vert S\right\rangle $ to decide whether $\left\vert
S\right\vert =w$\ or $\left\vert S\right\vert =2w$ with success probability at
least $2/3$, promised that one of those is the case.
\end{theorem}

\begin{proof}
Choose a subset $S\subseteq\left[  N\right]  $ uniformly at random, subject to
$\left\vert S\right\vert =w$\ or $\left\vert S\right\vert =2w$, and consider
$S$ to be fixed. \ Then suppose we choose $U\subseteq\left[  N\right]
$\ uniformly at random, subject to both $\left\vert U\right\vert =L$\ and
$S\subseteq U$. \ Consider the hybrid in which $Q$ is still given $R$\ copies
of the state $\left\vert S\right\rangle $, but now gets oracle access to
$\mathcal{O}_{U}$\ rather than $\mathcal{O}_{S}$. \ Then so long as $Q$ makes
$o\left(  \sqrt{\frac{N}{L}}\right)  $\ queries to its oracle, we claim that
$Q$ cannot distinguish this hybrid from the \textquotedblleft
true\textquotedblright\ situation (i.e., the one where $Q$\ queries
$\mathcal{O}_{S}$) with $\Omega\left(  1\right)  $\ bias. \ This claim follows
almost immediately from the BBBV Theorem \cite{bbbv}. \ In effect, $Q$ is
searching the set $\left[  N\right]  \setminus S$ for any elements of
$U\setminus S$\ (the \textquotedblleft marked items,\textquotedblright\ in
this context), of which there are $L-\left\vert S\right\vert $ scattered
uniformly at random. \ In such a case, we know that $\Omega\left(  \sqrt
{\frac{N-\left\vert S\right\vert }{L-\left\vert S\right\vert }}\right)
=\Omega\left(  \sqrt{\frac{N}{L}}\right)  $\ quantum queries are needed to
detect the marked items with constant bias.

Next suppose we first choose $U\subseteq\left[  N\right]  $ uniformly at
random, subject to $\left\vert U\right\vert =L$, and consider $U$ to be fixed.
\ We then choose $S\subseteq U$\ uniformly at random, subject to $\left\vert
S\right\vert =w$\ or $\left\vert S\right\vert =2w$. \ Note that this produces
a distribution over $\left(  S,U\right)  $\ pairs identical to the
distribution that we had above. \ In this case, however, since $U$ is fixed,
queries to $\mathcal{O}_{U}$\ are no longer relevant. \ The only way to decide
whether\ $\left\vert S\right\vert =w$\ or $\left\vert S\right\vert =2w$\ is by
using our copies of $\left\vert S\right\rangle $---of which, by assumption, we
need $\Omega\left(  k\right)  $\ to succeed with constant bias, even after
having fixed $U$.
\end{proof}

One might think that Theorem \ref{hybridthm} would lead to immediate
improvements to our lower bound. \ In practice, however, the best lower bounds
that we currently have, even purely on the number of copies of $\left\vert
S\right\rangle $, come from the Laurent polynomial method (Theorem
\ref{main})! \ Having said that, we are optimistic that one could obtain a
lower bound that beat Theorem \ref{main}\ at least when $w$\ is small, by
combining Theorem \ref{hybridthm} with a brute-force computation of trace distance.

\section{Discussion and Open Problems\label{OPEN}}

In Theorem \ref{main},\ can the bound $\min\left\{  w^{1/4},\sqrt{\frac{N}{w}%
}\right\}  $ be tightened to $\min\left\{  \sqrt{w},\frac{N}{w}\right\}  $,
matching the upper bounds that come from the birthday paradox and projective
measurements? \ If so, then how far can one go toward proving this using the
(Laurent) polynomial method,\ and where does one start to need new techniques?

Also, suppose our task was to distinguish the case $\left\vert S\right\vert
=w$\ from the case $\left\vert S\right\vert =\left(  1+\varepsilon\right)  w$,
rather than merely $w$\ from $2w$. \ Then what is the optimal dependence on
$\varepsilon$? \ As we said in Section \ref{INTRO}, it's known that $O\left(
\frac{1}{\varepsilon}\sqrt{\frac{N}{w}}\right)  $\ quantum queries to
$\mathcal{O}_{S}$\ suffice to solve this problem. \ One can also show without
too much difficulty that%
\[
O\left(  \min\left\{  \frac{\sqrt{w}}{\varepsilon},\frac{N}{\varepsilon^{2}%
w}\right\}  \right)
\]
copies of $\left\vert S\right\rangle $\ suffice. \ On the lower bound side,
what generalizations of Theorem \ref{main}\ can we prove that incorporate
$\varepsilon$? \ We note that our current argument doesn't
automatically\ generalize; one would need to modify something to continue
getting growth in the polynomials $u$\ and $v$\ after the first iteration.

Is there any interesting real-world example of a class of sets for which
QSampling and membership testing are both efficient, but approximate counting
is not? \ (I.e., the behavior that this paper showed can occur in the
black-box setting?)

Finally, our favorite open problem in this area: can we show that there's no
black-box $\mathsf{QMA}$\ protocol for approximate counting? \ In other words:
that there's no $\left(  \log N\right)  ^{O\left(  1\right)  }$-qubit quantum
state that Merlin can send to Arthur, so that Arthur becomes convinced after
$\left(  \log N\right)  ^{O\left(  1\right)  }$\ queries that $\left\vert
S\right\vert $\ is $2w$\ rather than $w$\ (promised that one of those is the
case)? \ Arthur's task is \textquotedblleft easier\textquotedblright\ than the
task considered in this paper, in that Merlin can send him an arbitrary
witness state $\left\vert \psi_{S}\right\rangle $, rather than just the
specific state $\left\vert S\right\rangle $; but also \textquotedblleft
harder,\textquotedblright\ in that Merlin can cheat and send the wrong state.
\ We thus obtain a problem that's formally incomparable to the one solved
here, yet which seems very closely related.

Ruling out a black-box $\mathsf{QMA}$\ protocol for approximate counting is
equivalent to asking for an oracle relative to which $\mathsf{SBP}%
\not \subset \mathsf{QMA}$, where $\mathsf{SBP}$ (Small Bounded-Error
Polynomial-Time), defined by B\"{o}hler et al. \cite{bohler}, is the
complexity class that captures the power of approximate counting. \ We note
that $\mathsf{MA}\subseteq\mathsf{SBP}\subseteq\mathsf{AM}$, and that an
oracle relative to which $\mathsf{AM}\not \subset \mathsf{QMA}$\ already
follows from the work of Vereshchagin \cite{vereshchagin}. \ We also note
that, under strong derandomization assumptions, we'd have $\mathsf{NP}%
=\mathsf{MA}=\mathsf{SBP}=\mathsf{AM}$, and hence $\mathsf{SBP}\subseteq
\mathsf{QMA}$\ in the unrelativized world.

\section{Acknowledgments}

I'm grateful to Paul Burchard\ for suggesting the problem to me, and to
\textquotedblleft fedja\textquotedblright\ for letting me include Lemmas
\ref{fedjalem}\ and \ref{fedjatight}.

\bibliographystyle{plain}
\bibliography{thesis}

\end{document}